\documentclass[10pt,conference,english,letterpaper]{IEEEtranTCOM}

\usepackage[T1]{fontenc}
\usepackage[fleqn]{amsmath}
\usepackage{graphicx}
\usepackage{amssymb}
\usepackage{enumerate}
\usepackage{paralist}
\usepackage{amsfonts}
\usepackage{amscd}
\usepackage{multirow}
\usepackage{color}
\usepackage{babel}
\usepackage{subfigure}
\usepackage{algorithm}
\usepackage{algorithmicx}
\usepackage{listings}
\usepackage{epsfig}
\usepackage{algpseudocode}
\usepackage{array}
\usepackage{amsmath}
\usepackage{perpage}
\MakePerPage{footnote}
\newtheorem{theorem}{Theorem}

\newtheorem{corollary}{Corollary}
\newtheorem{remark}{Remark}
\newtheorem{definition}{Definition}

\providecommand{\U}[1]{\protect\rule{.1in}{.1in}}
\makeatletter

\renewcommand{\figurename}{Fig.}
\addto\captionsenglish{\renewcommand{\figurename}{Fig.}}
\@ifundefined{definecolor}
{
}{}

\newtheorem{thm}{Theorem}[section]

\newtheorem{prop}[thm]{Proposition}

\makeatother

\usepackage[margin=19.5mm]{geometry}

\begin{document}
\IEEEoverridecommandlockouts

\author{
    \IEEEauthorblockN{Mustafa An{\i}l Ko\c{c}ak\IEEEauthorrefmark{1}, Elza Erkip\IEEEauthorrefmark{2}}
    \IEEEauthorblockA{Dept. of ECE, NYU Polytechnic School of Engineering Brooklyn, NY, USA \\
   kocak@nyu.edu\IEEEauthorrefmark{1}, elza@poly.edu\IEEEauthorrefmark{2}}
}

\title{\vspace{6.0mm} Communicating Lists Over a Noisy Channel}


\maketitle

\begin{abstract}
This work considers a communication scenario where the transmitter chooses a list of size $K$ from a total of $M$ messages to send over a noisy communication channel, the receiver generates a list of size $L$ and communication is considered successful if the intersection of the lists at two terminals has cardinality greater than a threshold $T$. In traditional communication systems $K=L=T=1$. 
The fundamental limits of this setup in terms of $K$, $L$, $T$ and the Shannon capacity of the channel between the terminals are examined. Specifically, necessary and/or sufficient conditions for asymptotically error free communication are provided.
\end{abstract}

\section{Introduction}

In the classical formulation of communication over a noisy channel \cite{Shannon}, the transmitter chooses a message from the set of possible messages, encodes and transmits the corresponding codeword over the channel. The receiver observes the channel output and decodes the chosen message. If the decoded message  is not same as the transmitted one, we declare an error. Alternatively, if the receiver cannot decide reliably the transmitted message, it can choose to produce a list of possible messages. This is known as \emph{list decoding}, where the main goal is to achieve a trade-off between the probability of error, defined as the transmitted message not being in the list, and the list size. An excellent treatment of list decoding can be found in \cite{forney} and references therein. Error exponent of list decoding is recently presented in \cite{merhav}, which also includes a survey of the literature.

In this work, we introduce a more general communication setup in which the transmitter chooses a subset of size $K$ of all possible $M$ messages. The receiver forms a list (set) size $L$ and communication is considered successful if the size of the intersection of the transmitter and receiver sets is larger than a certain threshold $T$. Note that $K=L=T=1$ corresponds to the classical approach whereas $K=T=1$ corresponds to the list decoding approach. 

This kind of communication may appear in various scenarios; for example in a wireless  network to alert a user about $K$ available resource blocks among $M$ possible ones, where the user is interested in a total of $T$ of the resource blocks and is willing to go through a list of $L$. Another possible scenario can be an internet search engine generating an unordered list of links size $K$ for a query which is transmitted over a noisy channel. The user is presented with a list of $L$ links and is satisfied as long as any $T$ of the $L$ results presented are relevant to his query. 

In this paper, we are interested in conditions on $M$, $K$, $L$ and $T$ that guarantee an asymptotically vanishing probability of communication failure. We investigate such necessary and sufficient conditions, and regions for which the necessary and sufficient conditions become equivalent, thereby providing tight bounds. We argue that the necessary and sufficient conditions developed in this paper depend on the channel only through its Shannon capacity. 

\subsection{Notation}
The following notation will be used throughout the paper:
\begin{itemize}
\item We use calligraphic capital letters for sets, and bold items for vectors.
\item For any positive integer $k$, $\left[k\right]$ stands for the set of all positive integers smaller or equal to $k$, i.e.  $\left[k\right] = \{1,\ldots,k\}$.
\item For any set $\mathcal{A}$, and any integer $k \leq \left|\mathcal{A}\right|$, ${\mathcal{A} \choose k}$ stands for the set of all size $k$ subsets of $\mathcal{A}$, i.e. $\binom{\mathcal{A}}{k}= \{ \mathcal{B}: \mathcal{B} \subseteq \mathcal{A}, \left|\mathcal{B}\right| = k\}$.
\item We denote the set of permutations of a given set $\mathcal{A}$ with $\Psi_\mathcal{A}$. Explicitly, any $\psi \in \Psi_\mathcal{A}$ is a bijective function from $\mathcal{A}$ to itself. The corresponding element of $a\in \mathcal{A}$ under this permutation is denoted by $\psi\left(a\right)$, and lastly for any $\mathcal{B} \subseteq \mathcal{A}$, $\psi\left(\mathcal{B}\right) = \{x : x = \psi\left(b\right), b\in \mathcal{B} \}$.
\item We use the standard notation for the equality and inequality up to first exponent, with respect to the natural base $e$, i.e. $f_n \dot{=} g_n$ stands for $\lim _{n\rightarrow\infty}\frac{1}{n}\log\frac{f_n}{g_n} = 0$, also $\dot{>}, \dot{<}, \dot{\geq}, \dot{\leq}$ are defined similarly as in \cite{cover}.
\end{itemize}

\section{Problem Definition}

The problem setup consists of two terminals and a channel in between characterized by  $\left(\mathcal{X},\mathcal{Y},\mathcal{W}\left(y^n|x^n\right)\right)$, where $\mathcal{X}$ denotes the input alphabet, $\mathcal{Y}$ denotes the output alphabet and $\mathcal{W}\left(y^n|x^n\right)$ is the transition probability of the channel for a block length $n$. In the following, we will denote the channel by $\left(\mathcal{X},\mathcal{Y},\mathcal{W}\right)$ for the sake of brevity. We assume the channel has a Shannon capacity denoted (in nats) by $C$. The first terminal (transmitter) chooses $K$ messages from the set of all possible messages, $\{1,\ldots,M\}$, and transmits a codeword of length $n$, $x^n \in \mathcal{X}^n$, through the channel. The second terminal (receiver), observes the channel output, $y^n\in \mathcal{Y}^n$, and generates a list of $L$ messages. Communication will be considered successful if the intersection of the estimated list and the set of chosen messages by the transmitter has cardinality larger than a given threshold $T$. We assume all the parameters $M,K,L,T$ are nondecreasing positive functions of block length $n$. 

The following definitions formalize the above setup.

\begin{definition}
An $\left(M,K,L,T,n\right)$ \em list encoding/decoding (LED) code \em for a given channel  $\left(\mathcal{X},\mathcal{Y},\mathcal{W}\right)$  consists of the following:
\begin{itemize}
\item Set of all possible messages: $\{1,\ldots, M\}$. 
\item Encoding function $f_n: {\left[M\right] \choose K}\longrightarrow \mathcal{X}^n$, which maps the chosen subset to channel inputs.
\item Decoding function $g_n: \mathcal{Y}^n \longrightarrow \binom{\left[\mathcal{M}\right]}{L}$, which generates the estimated message list from  channel output.
\item Overlap threshold: $T$, which is the success criterion for the designed code. 
\end{itemize}

\end{definition}

\begin{remark}
When $K = 1$, we call a LED code a \em list decoding (LD) code. \em When $K = L = 1$ we call a LED code a \em classical code.\em 
\end{remark}

\begin{definition}
An error is said to occur if and only if the size of the intersection of the decoded list and the
set of chosen messages is smaller than $T$. The probability of error when the message set $\Lambda \in {\left[M\right] \choose K}$ is sent can be expressed as:
\begin{eqnarray}
\lambda _ \Lambda ^{\left( n \right)} &=& 
 Pr\{\left| g_n\left(Y^n\right) \cap \Lambda \right| < T \quad  | X^n = f_n\left(\Lambda\right)\}. \nonumber
\end{eqnarray}

The average error probability is defined as

\begin{eqnarray}
  \lambda_{avg} ^{\left(n\right)} &=& \frac{1}{{M \choose K}}\sum_{\Lambda \in {\left[M\right] \choose K}} \lambda_\Lambda ^{\left(n\right)}. \nonumber
\end{eqnarray}

\end{definition}

\begin{definition}
A family of $\left(M,K,L,T,n\right)$ codes for a given channel $\left(\mathcal{X},\mathcal{Y},\mathcal{W}\right)$ is called a \emph{feasible family} if $\lambda_{avg}^{\left(n\right)}\rightarrow 0$ as $n \rightarrow \infty$.
\end{definition}

Our goal in this paper is to identify necessary and sufficient conditions for $\left(M,K,L,T,n\right)$ to obtain a feasible family of codes.

\section{Necessary and Sufficient Conditions for Feasibility}\label{sec3}

In this section, we present our main results. Specifically, Theorem 1 presents a combinatorial inequality which provides a sufficient condition on code parameters for the existence of a feasible family of codes.
Correspondingly, in Theorem 2 we describe a necessary condition that any feasible family of codes should satisfy, akin to converse results in classical coding. While our necessary and sufficient conditions are not tight in general, Section IV provides some important special cases for which the conditions become equivalent.

\begin{theorem}{(Sufficient condition for feasibility)} \label{thmsc}
There exists a feasible family of $\left(M,K,L,T,n\right)$ LED codes if 
\begin{eqnarray}
        \lim_{n \rightarrow \infty}\frac{1}{n} \log  \frac{\binom{M}{L}}{\sum_{i = T}^{\min\{K,L\}}\binom{K}{i}\binom{M-K}{L-i}} & < & C \nonumber
        \end{eqnarray}
where $C$ is the Shannon capacity of the  channel $\left(\mathcal{X},\mathcal{Y},\mathcal{W}\right)$.
\end{theorem}
\begin{proof}
We prove this result by constructing a fixed composition random code using a family of classical codes that achieves rate $R$ for the channel $\left(\mathcal{X},\mathcal{Y},\mathcal{W}\right)$. We denote this code family with $\{\left(\mathbf{u}_i,\mathcal{D}_i\right): i=1,\ldots,2^{nR}\}$, where $\mathbf{u}_i$'s are the codewords, $\mathcal{D}_i$'s are the decoding regions, and $R$ is the rate of the code \cite{cover}. Note that $R < C$.

\emph{Codebook Generation:} We pick $2^{nR}$ size $L$ subsets of $\left[M\right]$ i.i.d. with a uniform distribution over all subsets of size $L$ and denote them with $\{\mathcal{N}_1,\ldots,\mathcal{N}_{2^{nR}}\}$. Note these subsets are not necessarily distinct.  

\emph{Encoding:} For a given encoder input $\Lambda \in \binom{\left[M\right]}{K}$, we encode $\Lambda$ to  $\mathbf{u}_i$, such that $\left|\Lambda \cap \mathcal{N}_i\right| \geq T$. If there are more than one such $\mathcal{N}_i$'s we pick any one of them, and if such a $\mathcal{N}_i$ does not exists then we declare an encoding error denoted by the event $E_{enc}$.

\emph{Decoding:} At the receiver if the received channel output $y^n \in \mathcal{D}_j$, the receiver declares $\mathcal{N}_j$ as the estimated list. Note that if $y^n \in \mathcal{D}_j$ while the transmitted codeword is $\mathbf{u}_i$ for some $i\neq j$,
then we cannot guarantee $\left|\Lambda \cap \mathcal{N}_j\right| \geq T$ and we declare a decoding error. This event is denoted via $E_{dec}$.

\emph{Error Analysis:} The probability of error averaged over the ensemble of codes can be calculated as:
\begin{eqnarray}
E\left[\lambda_{avg}^{\left(n\right)}\right] & = &  E\left[\lambda_{\left[K\right]}^{\left(n\right)}\right]\label{three} \\ 
& \leq &  P\left(E_{dec} \right) + P\left(E_{enc} \right)  \label{four}\\
& = & P\left(E_{dec} \right) + \ldots \nonumber \\ & & +  \left(1 - \frac{\sum_{i = T}^{\min\{K,L\}}\binom{K}{i}\binom{M-K}{L-i}}{\binom{M}{L}}\right)^{2^{nR}} \nonumber
\\ & \leq &  P\left(E_{dec} \right) + e^{-2^{nR}  \frac{\sum_{i = T}^{\min\{K,L\}}\binom{K}{i}\binom{M-K}{L-i}}{\binom{M}{L}} } \label{six}
\end{eqnarray}
Eq. (\ref{three}) and (\ref{four}) follows from the symmetry of the code construction and the union bound, and (\ref{six}) follows from the simple inequality: $1-t \leq e^{-t}$ for all $t>0$.

Since the code $\{\left(\mathbf{u}_i,\mathcal{D}_i\right) : i = 1,\ldots, 2^{nR} \}$ is chosen to achieve rate $R$ for the classical communication problem,   $P\left(E_{dec} \right) \rightarrow 0$ as $n\rightarrow \infty$ for $R < C$. Furthermore, the second term also converges to zero as $n$ approaches to infinity, if 
\begin{eqnarray}
        \lim_{n \rightarrow \infty}\frac{1}{n} \log  \frac{\binom{M}{L}}{\sum_{i = T}^{\min\{K,L\}}\binom{K}{i}\binom{M-K}{L-i}} & < & R. \nonumber
        \end{eqnarray}
        
Thus, under the hypothesis of the theorem, we can always pick an appropriate $R$ to satisfy both of the conditions. Since the expected probability error over the ensemble can be made  asymptotically small, we conclude there exists at least one feasible code family in the ensemble with asymptotically zero error probability. 
\end{proof}

\begin{theorem}{(Necessary condition for feasibility)} \label{thmnc}
For any feasible family of $\left(M,K,L,T,n\right)$ LED codes for a channel $\left(\mathcal{X},\mathcal{Y},\mathcal{W}\right)$ with Shannon capacity $C$, we have
\begin{eqnarray}
\lim_{n \rightarrow \infty}\frac{1}{n} \log 
\frac{\binom{M}{K}/ \binom{L}{T}}{\sum_{i = T}^{K}\binom{K}{i}\binom{M-K}{K-i}} & \leq & C \nonumber
\end{eqnarray}

Furthermore, if $T = 1$, we have the tighter inequality
\begin{eqnarray}
\lim_{n \rightarrow \infty}\frac{1}{n} \log 
\frac{M}{KL} & \leq & C \nonumber
\end{eqnarray}

\end{theorem}
\begin{proof}
Suppose there exists a feasible family of $\left(M,K,L,T,n\right)$ LED codes for which
\begin{eqnarray}
\lim_{n \rightarrow \infty}\frac{1}{n} \log 
\frac{\binom{M}{K}/ \binom{L}{T}}{\sum_{i = T}^{K}\binom{K}{i}\binom{M-K}{K-i}} & > & C. \nonumber
\end{eqnarray}

Additionally assume, $\lim_{n \rightarrow \infty}\frac{1}{n} \log 
\frac{M}{KL}  >  C$ if $T=1$. We denote the encoding and decoding functions of the corresponding LED codes in this family as $\{f_n\}$ and $\{g_n\}$. Then one can construct a LD code 
as described below.

\emph{Code Generation:} 
\begin{enumerate}[i.]
\item Choose the largest set $\mathcal{P}  \subset \binom{\left[M\right]}{K}$ such that:
$\left| \mathcal{P}_i \cap \mathcal{P}_j \right| < T$ for any $\mathcal{P}_i, \mathcal{P}_j \in \mathcal{P}$, $i \neq j$. Note that, a direct application of Gilbert bound gives a lower bound on the size of $\mathcal{P}$ \cite{golomb}:
  \begin{eqnarray}
\left|\mathcal{P}\right| & \geq & \frac{\binom{M}{K}}{\sum_{i = T}^{K}\binom{K}{i}\binom{M-K}{K-i}}. \label{gilbert}
\end{eqnarray}
For $T = 1$, the cardinality of $\mathcal{P}$ can be calculated as $M/K$.
 
\item Generate a permutation, $\psi$ of $\left[M\right]$ randomly from the uniform distribution over  $\Psi_{\left[M\right]}$.
\item Fix the message set of the LD code as $\{1,\ldots, \left|\mathcal{P}\right|\}$.
\end{enumerate}

\emph{Encoding:} A message $i\in \{1,\ldots,\left|\mathcal{P}\right|\}$ is encoded as $f_n\left(\psi\left(\mathcal{P}_i\right)\right)$.

\emph{Decoding:} After observing the channel output $\mathbf{y}^n \in \mathcal{Y}^n$, the receiver generates the following list: 
\begin{eqnarray}
\mathsf{L} \left(\mathbf{y}^n\right) & = & \{i : \left|\psi ^{-1}\left(g_n\left(\mathbf{y}^n\right)\right) \cap \mathcal{P}_i\right|\geq T\}  \nonumber
\end{eqnarray}
Since the intersection of $\mathcal{P}_i,\mathcal{P}_j \in \mathcal{P}$, $i \neq j$ has size smaller than $T$ and $\left|\psi ^{-1}\left(g_n\left(\mathbf{y}^n\right)\right)\right|$ $=$  $L$ for all $\mathbf{y}^n$, we can conclude $\left|\mathsf{L}\left(\mathbf{y}^n\right)\right|\leq \binom{L}{T}$. For the case the list size is smaller than $\binom{L}{T}$, add $\binom{L}{T} - \left|\mathsf{L}\left(\mathbf{y}^n\right)\right|$ arbitrary messages to the list. Thus the final list size is fixed to $\binom{L}{T}$. 

\emph{Error Analysis:} We calculate the probability of error of the constructed LD code averaged both over the transmitted messages and the choice of $\psi$ by exploiting the symmetry introduced through the random permutation. If we denote the average probability of error for the LD code via $\lambda_{LD}^{\left(n\right)}$ and the average probability of error for the given $\left(M,K,L,T,n\right)$ LED code as $\lambda_{avg}^{\left(n\right)}$, we get
\begin{eqnarray}
E\left[\lambda_{LD}^{\left(n\right)}\right] & = & \lambda_{avg}^{\left(n\right)} \nonumber 
\end{eqnarray}
where the expectation is taken over the choice of $\psi$. 

Since $\lambda_{avg}^{\left(n\right)} \rightarrow 0$ as $n\rightarrow \infty$, we can conclude there exist at least one permutation  $\psi \in \Psi_{\left[M\right]}$ with corresponding error probability $\lambda_{LD}^{\left(n\right)} \rightarrow 0$. Note that the rate of the constructed LD codes are 
$R = \lim_{n\rightarrow \infty}\frac{1}{n}\log \left|\mathcal{P}\right|/\binom{L}{T}$ where $\left|\mathcal{P}\right|$ is the size of the message set and $\binom{L}{T}$ is the length of the list generated at the decoder. The converse result for LD codes indicates that, if $\lambda_{LD}^{\left(n\right)} \rightarrow 0$ as $n\rightarrow \infty$, then $R \leq C$ \cite{merhav, shgalber}.


Combining this with (\ref{gilbert}) we obtain:
\begin{eqnarray}
\lim_{n\rightarrow \infty}\frac{1}{n}\log \frac{\binom{M}{K}/ \binom{L}{T}}{\sum_{i = T}^{K}\binom{K}{i}\binom{M-K}{K-i}} &\leq & C. \nonumber
\end{eqnarray}
This leads to a contradiction and completes the proof.

For $T = 1$, we use $\left|\mathcal{P}\right|$ $=$ $M/K$ instead of the bound in (\ref{gilbert}) and the rest of the proof follows as above. 
\end{proof}

\begin{remark}
Theorem \ref{thmsc} and \ref{thmnc} suggest that the necessary and sufficient conditions presented depend on the channel only through its channel capacity $C$, and are applicable to all channels whose capacity can be determined. 
\end{remark}

\section{Feasibility, Rate and Capacity}
In this section we examine the asymptotic tightness of the bounds presented in the previous section under different regimes representing how $M,K,L,$ and $T$ increase with $n$. In particular, we assume 
 $K, L$ $\dot{<}$ $M$ and their exponential rates are finite, i.e. $\lim\limits_{n \rightarrow \infty}\frac{1}{n} \log M < \infty$.

 To simplify the presentation, we define the following quantities.
 
\begin{definition}
The \emph{rate} of a family of $\left(M,K,L,T,n\right)$ LED codes is defined by $\mathcal{R}$ $=$ $\lim_{n\rightarrow \infty}\frac{T}{n}\log \frac{MT}{KL}$.
\end{definition}

\begin{definition}
The {\em gap} for a family of $\left(M,K,L,T,n\right)$ LED codes is defined  by, 
\begin{eqnarray}
\mathcal{G} & = & \left\{
 \begin{array}{cl}
       0 & T = 1 \\ & \\
       \lim_{n\rightarrow \infty}\frac{T}{n}\log \frac{K}{T} & T > 1 \\
 \end{array} 
 \right.  \nonumber
 \end{eqnarray}
\end{definition}

\begin{corollary}
If $\mathcal{R} < C$  with $C > 0$ and finite, then there exists a feasible family of $\left(M,K,L,T,n\right)$ rate $\mathcal{R}$ LED codes for the  channel $\left(\mathcal{X},\mathcal{Y},\mathcal{W}\right)$ with capacity $C$.

\end{corollary}
\begin{proof}
The proof is given in two parts.

\begin{enumerate}[1.]
\item $0$ $\leq$  $\mathcal{R}$ $<$ $C$:

Note in this case, $MT$ $\dot{\geq}$ $KL$ since $\mathcal{R}$ $\geq$ $0$, and $T = O\left(1\right)$ since $C$ is finite.

We define $v_j = \binom{K}{j}\binom{M-K}{L-j}/\binom{M}{L}$, and consider the term $v_T$. Using Proposition \ref{prop2} in the Appendix:
\begin{eqnarray}
   v_T &\dot{=}& 
       \frac{1}{\sqrt{T}}\left(\frac{KLe}{MT}\right)^T e^{-\Delta_1} , \label{cor1} 
       \end{eqnarray} 
       where  $\Delta_1  = \Theta\left(\frac{T^2}{K} + \frac{T^2}{L}+\frac{KL}{M}\right)$. 
       
       Since $T$ is bounded and $M$ $\dot{>}$ $KL$,
       \begin{eqnarray}
          v_T &\dot{=}& 
              \left(\frac{KL}{M}\right)^T.     \nonumber         
              \end{eqnarray} 
       The hypothesis of the corollary implies:
       \begin{eqnarray}
e^{-nC} &\dot{<} & v_T \nonumber \\
& \dot{<} & \sum\limits_{j=T}^{\min \{K,L\}}v_j. \label{cor33}
       \end{eqnarray}
   Finally (\ref{cor33}) and Theorem \ref{thmsc} guarantee the existence of a feasible family of codes.

\item $\mathcal{R}$ $<$ $0$:

For any  $j_0$ $=$ $\frac{KL}{M}e^{\alpha n}$ with some nonnegative  $\alpha$, such that $0$ $<$ $j_0$ $\dot{<}$ $K$ and $0$ $<$ $j_0$ $\dot{<}$ $L$, Proposition \ref{prop2} suggests that:
\begin{eqnarray}
   v_{j_0} &\dot{=}& 
       \frac{1}{\sqrt{j_0}}e^{\left(1-\alpha n\right)e^{\alpha n}\frac{KL}{M}} e^{-\Delta_2},
         \label{cor2}
       \end{eqnarray}
       where $\Delta_2 = \frac{KL}{M}\left(1+ o\left(1\right)\right)$.

This implies that if $\alpha$ is non-zero, then
\begin{eqnarray}
\lim_{n\rightarrow \infty} \frac{1}{n}\log v_{j_0} & = & \lim_{n\rightarrow \infty} \frac{\left(1-\alpha n - e^{-\alpha n}\right)j_0}{n} \nonumber \\
& = & -\infty. \label{mi}
\end{eqnarray}

Also note the following always holds under the condition $\mathcal{R} < 0$,
\begin{eqnarray}
\sum_{j=T}^{\min \{K,L\}} v_j & = & 1 - \sum_{j=0}^{T-1} v_j \label{id1} \\
& \dot{=} &  1 \quad  \dot{>} \quad e^{-nC}, \label{id2} 
\end{eqnarray}
for any $C>0$.
Here (\ref{id1}) follows from the identity $\sum_{j=0}^{\min \{K,L\}} v_j  =  1$, and (\ref{id2}) follows from the fact that all the terms in the sum decay faster than any first order exponential due to (\ref{mi}).

Finally, (\ref{id1}), (\ref{id2}) and Theorem \ref{thmsc} imply the existence of a feasible family of $\left(M,K,L,T,n\right)$ codes.

\end{enumerate}

\end{proof}

\begin{remark}
Unlike the rate of classical codes, the rate of LED codes may be negative.  When the rate is negative, which could happen for example if the list sizes at the encoder and/or decoder grow fast enough, we have feasibility for any channel whose Shannon capacity is non-zero.
\end{remark}

\begin{corollary}
If there exists a feasible family of $\left(M,K,L,T,n\right)$ rate $\mathcal{R}$ LED codes for a channel $\left(\mathcal{X},\mathcal{Y},\mathcal{W}\right)$ with capacity $C$, then $\mathcal{R} - \mathcal{G} \leq C$.
\end{corollary}

\begin{proof}
Note for $T$ $=$ $1$, the proof immediately follows from Theorem \ref{thmnc}. Also for $MT$ $\dot{\leq}$ $K^2$, $\mathcal{R} - \mathcal{G}$ $\leq$ $0$, and the Corollary follows. Hence for the rest of the proof we assume $MT$ $\dot{>}$ $K^2$.

First define 
\begin{eqnarray}
J  & = & \left\{
	\begin{array}{ll}
		e^{\alpha n} & \mbox{if } K \dot{>} 0 \\
		K & \mbox{otherwise }
	\end{array} \nonumber
\right.
\end{eqnarray}
for any $\alpha$ such that $K$ $\dot{\geq}$ $J$ $\dot{\geq}$ $T$.

Then Theorem \ref{thmnc} suggests
\begin{eqnarray}
 e^{-nC} & \dot{\leq} & \binom{L}{T}\sum\limits_{j=T}^{K}w_j  \nonumber
\end{eqnarray}
where $w_j$ $=$ $\frac{\binom{K}{j}\binom{M-K}{K-j}}{\binom{M}{K}}$. By Proposition \ref{prop1} and the assumption $MT$ $\dot{>}$ $K^2$, $w_j$ is a decreasing sequence for large $n$. Therefore
\begin{eqnarray}
 e^{-nC} & \dot{\leq} & \binom{L}{T}\sum\limits_{j=T}^{J}w_j + \binom{L}{T}\sum\limits_{j=J}^{K}w_j \nonumber \\
  & \dot{\leq} & \binom{L}{T}J w_T + \binom{L}{T}Kw_{J}. \nonumber
\end{eqnarray}
By applying Proposition \ref{prop2} and (\ref{bin}), (\ref{expo}) in the Appendix, we get
\begin{eqnarray}
 e^{-nC} & \dot{\leq} & \frac{J}{T}\left(\frac{K^2L}{MT^2}\right)^T + \frac{KL^T}{T}\left(\frac{K^2}{MT^2}\right)^{J} \label{cor3inn}
\end{eqnarray}
Note the second term in the sum in (\ref{cor3inn}) has a double exponential decay, thus:
\begin{eqnarray}
 e^{-nC} & \dot{\leq} & \frac{J}{T}\left(\frac{K^2L}{MT^2}\right)^T.\label{cor3in}
\end{eqnarray}
\begin{enumerate}[i.]
\item If $T$ $=$ $O\left(1\right)$,  then we can choose $\alpha$ arbitrarily small and  (\ref{cor3in}) simplifies to:
\begin{eqnarray}
 e^{-nC} & \dot{\leq} & \left(\frac{K^2L}{MT^2}\right)^T. \nonumber
\end{eqnarray}
which leads to $\mathcal{R}-\mathcal{G} \leq C$.

\item Otherwise for $MT^2 \dot{>} K^2L$,  
\begin{eqnarray}
\lim\limits_{{n\rightarrow \infty}} \frac{1}{n}\log {\frac{J}{T}\left(\frac{K^2L}{MT^2}\right)^T} &=& -\infty \nonumber
\end{eqnarray}
 and (\ref{cor3in})  cannot be true. This implies that if $T \rightarrow \infty$, then $MT^2$ $\dot{\leq}$ $K^2L$  and $\mathcal{R}-\mathcal{G} \leq 0$.
\end{enumerate}
\end{proof}

\begin{remark}
Corollary 1 and 2 suggest that the necessary and sufficient conditions of Section \ref{sec3} become tight for the special case $T = 1$. Hence for $T = 1$, the existence of a feasible family of LED codes can be determined by simply comparing the code rate (as in Definition 4) with the Shannon capacity of the channel. 

Moreover, for feasibility the list sizes at the terminals are transferable  as long as their product is conserved, since the rate only depends on the product $KL$.
\end{remark}

\section{Conclusions}
In this paper, we have studied how to communicate a chosen subsets of messages over a noisy channel. In our treatment, we have provided fundamental limits for feasibility in terms of number of messages $\left(M\right)$,  the list size at transmitter $\left(K\right)$, at the receiver $\left(L\right)$, overlap threshold $\left(T\right)$, and the Shannon capacity of the channel $\left(C\right)$.  
Specifically, we have derived necessary and sufficient conditions for asymptotically error free communication and argued that they depend on the channel only through its capacity. Our results have illustrated that for the special case $T = 1$ the necessary and sufficient conditions coincide and communication is possible only when rate, defined as  $\mathcal{R} = \lim_{n\rightarrow \infty}\frac{1}{n}\log \frac{M}{KL}$ is smaller than capacity. 


\appendix

In this appendix we examine the asymptotic behavior of the sequence $v_j = \binom{K}{j}\binom{M-K}{L-j}/\binom{M}{L}$ for $j=0,\ldots, \min \{K,L\}$. For the following we assume, as in the previous sections, $M$ exponentially increases with $n$, $M\dot{>} K$ and $M \dot{>} L$.

\begin{prop}\label{prop1}
There exists a $n_0 \in \mathbb{N}$ such that the sequence $\{v_j\}$ is either unimodal or decreasing in $j$ for all $n > n_0$. If the sequence is unimodal, its maximum is reached for some $j^*$ $\dot{=}$ $KL/M$. 
\end{prop}
\begin{proof}
First define,
\begin{eqnarray}
a_j & = & \frac{v_{j+1}}{v_j} \nonumber \\
 & = & \frac{\left(K-j\right)\left(L-j\right)}{\left(M-K-L+j+1\right)\left(j+1\right)}. \nonumber
\end{eqnarray}
Since $\{a_j\}$ is a decreasing sequence in $j$, and $a_s = 0$ for $s$ $=$ $\min\{K,L\}$; $v_j$ is decreasing if $a_0$ is smaller than $1$ and unimodal if $a_0$ is larger than $1$.  

Considering the following limits:
\begin{eqnarray}
\lim\limits_{n \rightarrow \infty} a_0 & = &  \lim\limits_{n\rightarrow \infty} \frac{KL}{M} \nonumber  \\ 
\lim\limits_{n \rightarrow \infty} a_{j_0} & = & \lim\limits_{n\rightarrow \infty} \frac{KL}{Mj_0} \quad = \quad 1  \nonumber \\
\lim\limits_{n \rightarrow \infty} a_{j_1} & = & \lim\limits_{n\rightarrow \infty} \frac{j_0}{j_1} \qquad = \quad 0 \nonumber 
\end{eqnarray}
for any $j_0$  $\dot{=}$  $\frac{KL}{M}$ and $j_1$ $\dot{>}$ $\frac{KL}{M}$.

We can argue the existence of $n_0$ such that for all $n > n_0$,  $v_j$ is decreasing if $M \dot{>} KL$, and unimodal with the maximum at some $j^*$ $\dot{=}$ $j_0$ otherwise.
\end{proof}

\begin{prop}\label{prop2}
For any $0$ $<$ $j$ $<$ $K$ and $j$ $<$ $L$: 
\begin{eqnarray}
v_j & \dot{=} & \frac{1}{\sqrt{j}}\left(\frac{KLe}{Mj}\right)^je^{-\Delta} \qquad \textrm{where} \nonumber
\end{eqnarray} 
$
\Delta$ $=$ $\sum\limits_{k=1}^{\infty} \frac{1}{k\left(k+1\right)}\left(\frac{\left(K+L-j\right)^{k+1} - K^{k+1} - L^{k+1}}{M^k}+ \frac{j^{k+1}}{K^k} + \frac{j^{k+1}}{L^k}\right)$.

\end{prop}

\begin{proof}
For any $A,B,C \in \mathbb{N}$:
\begin{eqnarray}
\binom{A}{B} & = & \frac{A^{A+1/2}\quad  e^{O\left(\frac{1}{B} + \frac{1}{A-B}\right)}}{B^{B+1/2}\left(A-B\right)^{A-B+1/2}}, \label{bin} 
\end{eqnarray}
\begin{eqnarray}
\left(1-\frac{A}{B}\right)^C & = & e^{C\log \left(1 - A/B\right)} =  e^{-C \sum\limits_{k=1}^{\infty}\frac{A^k}{B^k}}. \label{expo}
\end{eqnarray}
Here (\ref{bin}) follows from successive application of Sterling approximation \cite{asym}, and (\ref{expo}) follows from the Taylor expansion of $\log \left(1-x\right)$. 

The proposition simply follows from application of those identities on $v_j$. By noting  $\left(1-\frac{A}{B}\right)^C$ $\dot{=}$ $1$ for $B$ $\dot{>}$ $AC$, we have 
\begin{eqnarray}
v_j & \dot{=} & \frac{1}{\sqrt{j}}\left(\frac{KL}{Mj}\right)^j \times \ldots \nonumber \\
 & &\times   \frac{\left(1-\frac{K}{M}\right)^{M-K}\left(1-\frac{L}{M}\right)^{M-L}}{\left(1-\frac{K+L-j}{M}\right)^{M-K-L+j}\left(1-\frac{j}{K}\right)^{K-j}\left(1-\frac{j}{L}\right)^{L-j}} \nonumber \\
& \dot{=} & \frac{1}{\sqrt{j}}\left(\frac{KL}{Mj}\right)^j e^{j-\Delta} \quad = \quad  \frac{1}{\sqrt{j}}\left(\frac{KLe}{Mj}\right)^j e^{-\Delta}. \nonumber
\end{eqnarray}
\end{proof}



\begin{thebibliography}{1}

\bibitem{Shannon}
C. E. Shannon, ``Communication in the presence of noise.'' \emph{Proceedings of the IRE} 37.1 (1949): 10-21.


\bibitem{forney}
G. D. Forney Jr, ``Exponential error bounds for erasure, list, and decision feedback schemes.'' \emph{Information Theory, IEEE Transactions on} 14.2 (1968): 206-220.


\bibitem{merhav}
N. Merhav, ``List decoding-random coding exponents and expurgated exponents.'' \emph{arXiv} preprint arXiv:1311.7298 (2013).


\bibitem{cover}
T. M. Cover, and J. A. Thomas, \emph{Elements of Information Theory}. John Wiley \& Sons, 2012.


\bibitem{golomb}
S. W. Golomb, R. E. Peile, and R. A. Scholtz, \emph{Basic Concepts in Information Theory and Coding: The Adventures of Secret Agent 00111.} Springer, 1994.



\bibitem{shgalber}
C. E. Shannon, R. G. Gallager, and E. R. Berlekamp, ``Lower bounds to error probability for coding on discrete memoryless channels. I.'' \emph{Information and Control} 10.1 (1967): 65-103.


\bibitem{asym}
A. M. Odlyzko, ``Asymptotic enumeration methods.'' Handbook of Combinatorics 2 (1995): 1063-1229.


\end{thebibliography}
\end{document}